\documentclass[a4paper,UKenglish,cleveref, autoref, thm-restate]{lipics-v2021}

\usepackage{wrapfig}
\usepackage{enumitem}
\usepackage{graphicx}

\usepackage[linesnumbered,lined,boxed,commentsnumbered]{algorithm2e}
\usepackage{multirow}
\usepackage{tabularx}
\usepackage{tikz}
\usepackage{tikzscale}
\usetikzlibrary{calc}

\usepackage[]{color}

\title{Improved FPT Approximation for Non-metric TSP} 

\author{Evripidis Bampis}{Sorbonne Universit\'e, CNRS, LIP6, F-75005 Paris, France}{evripidis.bampis@lip6.fr}{}{}

\author{Bruno Escoffier}{Sorbonne Universit\'e, CNRS, LIP6, F-75005 Paris, France \and Institut Universitaire de France, Paris, France}{bruno.escoffier@lip6.fr}{}{
}

\author{Michalis Xefteris}{Sorbonne Universit\'e, CNRS, LIP6, F-75005 Paris, France}{michail.xefteris@lip6.fr}{}{
}

\authorrunning{E. Bampis, B. Escoffier and M. Xefteris} 

\Copyright{Evripidis Bampis, Bruno Escoffier and Michalis Xefteris} 



\category{} 

\relatedversion{} 

\funding{}

\acknowledgements{
}

\nolinenumbers 

\newcommand{\opt}{\ensuremath{OPT}}
\newcommand{\alg}{\ensuremath{ALG}}

\begin{document}

\maketitle

\begin{abstract}
  In the Traveling Salesperson Problem (TSP) we are given a list of locations and the distances between each pair of them. The goal is to find the shortest possible tour that visits each location exactly once and returns to the starting location. Inspired by the fact that general TSP cannot be approximated in polynomial time within any constant factor, while metric TSP admits a (slightly better than) $1.5$-approximation in polynomial time,  Zhou, Li and Guo~\cite{zhou} introduced a parameter that measures the distance of a given TSP instance from the metric case. They gave an FPT $3$-approximation algorithm parameterized by $k$, where $k$ is the number of triangles in which the edge costs violate the triangle inequality. In this paper, we design a $2.5$-approximation algorithm that runs in FPT time, improving the result of~\cite{zhou}.
\end{abstract}

\section{Introduction}
The Traveling Salesperson Problem (TSP) is a classic optimization problem in the field of computer science and operations research.  In (symmetric) TSP, we are given a complete undirected graph $G = (V, E)$ with $n$ vertices and an edge cost function $c: E \rightarrow \mathbb R_{\ge 0}$. The goal is to find a minimum-cost cycle visiting every vertex exactly once (Hamiltonian cycle).  TSP has many applications in the logistics sector, particularly in optimizing and planning delivery routes. Algorithms that solve TSP instances aid in reducing both travel time and costs. Unfortunately, TSP is an NP-hard problem~\cite{Karp, garey}; it is solvable in exponential time (in time $O(n^2 \cdot 2^n)$ by algorithms proposed independently by Bellman~\cite{bellman} and Held and Karp~\cite{held}), but not in subexponential time $2^{o(n)}$ under ETH~\cite{ImpagliazzoP01}.

With respect to (polytime) approximation algorithms, TSP does not admit any constant ratio approximation algorithm unless $P=NP$~\cite{teofilo}, but some classical subcases are well known to be easier to approximate. In particular, in the metric case (i.e., when costs of edges satisfy the triangle inequality $c(u,v)\leq c(u,w)+c(w,u)$), there exists a $3/2$-approximation algorithm in $O(n^3)$ time given independently by Christophides~\cite{Christofides} and Serdyukov~\cite{Ser78}. This ratio of 3/2 has been recently slightly improved~\cite{random_tsp, det_tsp}.  On the other hand, metric TSP does not admit a PTAS unless $P=NP$~\cite{PapadimitriouY93}. Arora~\cite{arora} and Mitchell~\cite{Mitchell} have shown that if we further restrict the instances to the Euclidean case, where vertices are points in the plane (or more generally points in a $d$-dimentional space) and distances are Euclidean distances, then the problem (known to remain $NP$-hard) admits a PTAS. The asymmetric version of the problem (i.e., TSP in directed complete graphs, where $c(u,v)$ may be different from $c(v,u)$) has been also widely studied. The question whether metric ATSP admits a constant approximation algorithm had been open for many years, and recently got a positive answer~\cite{SvenssonTV20}.  


In an attempt to generalize (symmetric) metric TSP instances to non-metric instances, two approaches have been mainly considered. The first one is to consider that only a relaxed version of the triangle inequality holds. Andreae and Bandelt~\cite{andreae} proposed to relax the triangle inequality in the following way: for some $\tau \ge 1$, we require that $c(u,v) \le \tau[c(u,w)+c(w,v)]$ for every $u,v,w \in V$. For these instances of TSP, they proposed a $(3\tau^2+\tau)/2$-approximation algorithm. Later, Bender and Chekuri~\cite{bender} improved this ratio for $\tau>7/3$ with a ratio of $4\tau$, and showed that getting ratios sublinear in $\tau$ is impossible unless $P=NP$. The ratio has been further improved for small values of $\tau$ down to $3(\tau+\tau^2)/4$ in~\cite{MOMKE2015866}. 

A second approach is to consider that triangle inequality is only verified in some parts of the instances. In that vein, Mohan~\cite{mohan} designed a $7/2$-approximation algorithm for a version of TSP where all vertices can be partitioned into two
parts $V_1,V_2$ such that the induced subgraphs $G[V_1]$ and $G[V_2]$ are metric. A recent work~\cite{zhou} defines some parameters that measure how large is the ``non-metric'' part of the instance, and then proposes  {\it parameterized approximation algorithms.} Parameterized complexity is a classic approach to solve $NP$-hard problems, the goal of which is to express the complexity as a function of both the instance size $|I|$ and a parameter $k$ on the instance. An FPT (for Fixed-Parameter Tractable) algorithm is an algorithm whose time complexity is upper bounded by $f(k)|I|^{O(1)}$. We refer the reader to some textbooks on parameterized complexity such as~\cite{FlumG06,CyganFKLMPPS15}. 
For some problems, there exist exact algorithms with FPT time complexity, but for many others this is not possible, under plausible complexity theoretic assumptions. A complementary approach to address this intractability is to study parameterized approximation algorithms.
A parameterized $\rho$-approximation algorithm is a $\rho$-approximation  algorithm that runs in $f(k) \cdot |I|^{O(1)}$ time. Many problems have been studied under the lens of parameterized approximability. We refer the reader to the surveys for parameterized approximation algorithms of~\cite{marx} and~\cite{Feldmann}. The aforementioned PTAS of Arora~\cite{arora} in $d$-dimensional metric space can be seen as a parameterized PTAS, where the parameter is the dimension $d$ of the space, as it gives a $(1+\epsilon)$-approximation algorithm for metric TSP in time $f(d,\epsilon) \cdot n^2$ for some function $f$. 
Another example is the work by Bockenhauer et al.~\cite{Bockenhauer} who studied TSP with deadlines, in which a subset of $k$ vertices must be visited before a given time. They designed a parameterized $2.5$-approximation algorithm for the metric case running in $O(k! \cdot k + n^3)$ time.

As mentioned above, and very closely related to our work, Zhou, Li and Guo~\cite{zhou} introduced two parameters that measure the so-called “distance from approximability” for TSP. In TSP, the distance between the inapproximable and approximable case “translates” into the distance
between a general TSP instance and a metric TSP instance. First, they considered the number $k$ of triangles in a general TSP instance, which do not satisfy the triangle inequality (the edge costs of the triangle), and gave a $3$-approximation algorithm in time $O\big((3k)! 8^k \cdot n^2+n^3\big)$. Additionally, they introduced a second parameter $k'$ which measures the minimum number of vertices whose removal turns a general instance into a metric instance. They designed a $(6k'+9)$-approximation algorithm that runs in time $O(k'^{O(k')} \cdot n^3)$ and a $3$-approximation algorithm that runs in  $O(n^{O(k')})$ time. Note that it holds that $k' \le k$.
 
\paragraph*{Our contribution}
In this paper, we revisit the problem of designing parameterized approximation algorithms for TSP using a parameter that measures the distance between a general TSP instance and a metric one, as introduced in~\cite{zhou}. More specifically, we consider the parameter $k$ which measures the number of triangles that violate the triangle inequality, and design an improved $2.5$-approximation algorithm that runs in time $O\big((3k)! 8^k \cdot n^4\big)$.

 Our algorithm first follows almost the same steps as in the $3$-approximation algorithm of Zhou, Li and Guo~\cite{zhou}. Roughly speaking, they build a specific connected spanning subgraph of the initial graph, and then double some edges to get an Eulerian graph. From an Euler tour of this graph they build a Hamiltonian cycle. We point out a crucial difficulty here: as the graph is not metric, transforming an Euler tour into a Hamiltonian cycle may increase significantly its cost; to avoid this, shortcutting should only be done on triangles satisfying triangle inequality. In our algorithm, instead of doubling the edges to create an Eulerian graph, we compute a minimum-cost perfect matching of a specific subgraph, drawing ideas from the Christofides' algorithm~\cite{Christofides}. We point again a difficulty due to non metricity: the usual bound on the matching cost does not directly follows, as in non metric instances the best Hamiltonian cycle of a subgraph can be larger that the one of the initial graph. Adding the matching to the spanning subgraph gives an Eulerian graph. To avoid the first shortcutting problem mentioned above,
we have an additional step between the construction of the Euler tour and shortcutting vertices that are visited more than once. This step then allows us to obtain a Hamiltonian cycle by carefully handling the  shortcutting phase. 

We also note that our algorithm (like the algorithm in~\cite{zhou}) runs in FPT time $O(k_T!2^{k_T}n^4)$ for the more intuitive parameter $k_T$ equal to the number of vertices that appear in at least one triangle where the triangle inequality is violated. This parameter can be much smaller than the number of the violating triangles. In what follows, we first introduce some notation in Section~\ref{preliminaries}, and then give the description of the algorithm and the analysis of its time complexity and approximation ratio (Section~\ref{analysis}).

\section{Preliminaries} \label{preliminaries}
In this section, we give some basic definitions and notation. We follow the notation of~\cite{zhou}. Throughout the paper, we consider a simple undirected complete graph $G(V,E)$ with an edge cost function $c:E \rightarrow \mathbb R_{\ge 0}$. If for every $u, v, w \in V$ it holds that $c(u,v) \le c(u,w)+c(v,w)$, then the graph $G$ is called metric. When we refer to a violating triangle $\Delta(u,v,w)$, we mean that at least one of $c(u,v) \le c(u,w)+c(v,w)$, $c(u,w) \le c(u,v)+c(v,w)$ and $c(v,w) \le c(v,u)+c(w,u)$ does not hold. For an edge subset $E' \subseteq E$, the cost of $E'$ is the total sum of the cost of its edges and is denoted by $c(E')$. The vertex set of $E'$ is denoted by $V(E')$. For a vertex subset $V' \subseteq V$, $E(V')$ is the set of edges that connects two vertices in $V'$, and $G[V']=(V',E(V'))$ is the subgraph of $G$ induced by $V'$. For a positive integer $i$, set $[i]= \{ 1, 2,\dots,i\}$.

A $t$-forest is an acyclic graph consisting of $t$ disjoint trees. When a spanning subgraph of $G$ is a $t$-forest, it is called a spanning $t$-forest of $G$. A spanning $t$-forest that has the minimum cost is a $t$-minimum spanning forest ($t$-MSF) of $G$.

In this work, we study the TSP problem in which $k$ triangles violate the triangle inequality. This parameter $k$ can be easily computed in $O(n^3)$ time by considering all triangles of the input graph. We call a vertex “bad” if it is contained in at least one of the $k$ violating triangles, and we denote by $V^b$ the set of all bad vertices, with $|V^b| \le 3k$. The rest of the vertices are called “good” and the set of them is denoted by $V^g$. For $b_1, b_2, b_3 \in V^b$, the triangle formed by these three vertices $\Delta(b_1,b_2,b_3)$ may violate the triangle inequality. For example, $c(b_1,b_3) + c(b_2,b_3)$ can be arbitrary less than $c(b_1,b_2)$. Additionally, in any triangle $\Delta(g,u,v)$ that contains at least one “good” vertex $g \in V^g$, with $u,v \in V$, all triangle inequalities between $g, u, v$  hold.

For the rest of the paper, we denote the length of a TSP tour output by an algorithm $\alg$ by $c(\alg)$ and that of an optimal (minimum length) TSP solution $\opt$ by $c(\opt)$.

\section{Parameterized Approximation for TSP} \label{analysis}
In this section, we will describe our $2.5$-approximation algorithm for TSP in FPT time, parameterized by the number of violating triangles $k$. Our algorithm first follows the initial steps of the $3$-approximation algorithm of~\cite{zhou}, and then carefully constructs a Hamiltonian cycle of all $n$ vertices. Let us first give a rough description of the main steps of our algorithm, which will be given in more details in Section~\ref{subsec:4.1}.

As in~\cite{zhou}, the first step involves "guessing" the positions of bad vertices in an optimal TSP solution, specifically their order and the "gaps" between them where good vertices should be placed. This is achieved by enumerating all possible permutations of the bad vertices and, for each permutation, all partitions that respect the order of the permutation. In other words, the bad vertices in each subset of a partition appear together in the optimal TSP solution and follow the order dictated by the permutation. This means that the bad vertices in one subset of the partition appear together in the optimal TSP solution and respect the order of the permutation. Additionally, the subsets occur in the order of the permutation, and at least one good vertex is contained between any two consecutive subsets. Then, we compute a minimum-cost spanning forest of only good vertices rooted at the (bad) end-vertices of the subsets of the partition. In this way, we fill in the gaps between the subsets in each partition of permutations with good vertices. A new graph $G'$ is created, that contains (1) a Hamiltonian cycle of bad vertices which follows the order in which an optimal TSP visits them, and (2) the minimum-cost spanning forest that has been computed. 
Next, following the general idea of the Christofides' algorithm we find a minimum-cost perfect matching on the vertices of odd degree in the built graph $G'$. By adding this perfect matching to $G'$ we get a graph $H$ with only even-degree vertices. We then make some local modifications, without increasing the total cost of edges, that will eventually allow us to build a Hamiltonian cycle the cost of which is not larger than the one of an Euler tour of $H$.


We introduce some additional notation, following the ones in~\cite{zhou}, that will help us describe our algorithm. A bad chain is one bad vertex or a path consisting of distinct bad vertices. It is denoted by $q=(b_1,b_2,\dots,b_{\ell})$ with $\ell \ge 1$, where $b_i \in V^b$ for $i \in [\ell]$ and there exists an edge connecting $b_i$ and $b_{i+1}$ for each $i \in [\ell-1]$.
We use $b_s(q)$, $b_e(q)$ to refer to the starting and ending vertices, respectively, and use $c(q)$ to
denote the total cost of the edges in $q$.

Moreover, an optimal TSP solution $\opt$ can be decomposed into an ordered collection of $2 t^{opt}$ many paths $q_1^{opt}, p_1^{opt}, q_2^{opt}, p_2^{opt}, \dots, p_{t^{opt}}^{opt}$ with $q_i^{opt}$ being a bad chain and $p_i^{opt} = (b_e(q_i^{opt}), g_{i_1}^{opt}, \dots, g_{i_{\ell_i^{opt}}}^{opt}, b_s(q_{i+1}^{opt}))$ being the path connecting $b_e(q_i^{opt})$ and $b_s(q_{i+1}^{opt})$ with only good internal vertices. Here $\ell_i^{opt} \ge 1$ and $q_{t^{opt}+1}^{opt}= q_1^{opt}$. Of course $t^{opt} \le 3k$.

\subsection{Algorithm and Time Complexity}\label{subsec:4.1}
We are now ready to give the description of Algorithm~\ref{algo}. The algorithm uses \textsc{ShortCut} (see Algorithm~\ref{shortcut}) as a subroutine  (see Figure~\ref{fig:algorithm-1} and~\ref{fig:shortcut} for a graphic illustration). Note that we assume that there exists at least one good vertex in graph $G$, otherwise the problem becomes trivial and Algorithm~\ref{algo} outputs after Step~2(a) an optimal TSP solution.

\begin{algorithm}[ht]
	\KwIn{$G(V,E)$ with cost function $c:E \rightarrow \mathbb R_{\ge 0}$}
    \begin{enumerate}
        \item Compute all $k$ triangles that violate the triangle inequality, the set of bad vertices\\ $V^b$ and the set of good vertices $V^g$.

        \item Enumerate all possible permutations of bad vertices. For each permutation,\\ enumerate all possible partitions of bad vertices into $t$ subsets for each $t \in [|V^b|]$, respecting the corresponding order of permutation.\\ For each $t$-partition, do the following:

        \begin{enumerate}[label=(\alph*)]
            \item Connect the bad vertices in each subset of the partition in the permutation\\ order, creating $t$ bad chains $Q=(q_1,\dots,q_t)$ according to the order of the permutation. Additionally, connect $q_i$ with $q_{i+1}$
            with edge $(b_e(q_i),b_s(q_{i+1}))$,\\ for each $i \in [t]$ with $q_{t+1}=q_1$, resulting in one simple cycle  $C$ of all bad vertices.

            \item Compute a $t$-minimum spanning forest ($t$-MSF) $F=(T_1,\dots,T_t)$ of  $G[V^g \cup \{b_e(q_i)|i \in
             [t]\}]$ rooted at $\{b_e(q_i)|i \in [t]\}$.

             \item Combine the edges of $C$ and $F$ to form a connected graph $G'$ of all vertices in\\ $V$. Let $O$ be the set of vertices with odd degree in $G'$. 

            \item Find a minimum-cost perfect matching $M$ in the subgraph induced in $G$ by $O$.

            \item Combine the edges of $G'$ and $M$ to form a connected multigraph $H$ in which\\ each vertex has even degree.

            \item Call \textsc{ShortCut}($H$) and output a TSP tour of $G$.
        \end{enumerate}
    \end{enumerate}
    \caption{$2.5$-approximation Algorithm for TSP parameterized by $k$}
    \label{algo}
\end{algorithm}

\begin{algorithm}[ht]
	\KwIn{Multigraph $H$ with vertices $V=V^g \cup V^b$}
    \begin{enumerate}
        \item While there exists a double edge between vertices $b_e(q_i^{opt})$ and $b_e(q_{i+1}^{opt})$ for $i \in [t^{opt}]$\\ with $q_1 = q_{t{opt}}$, do the following: 

        \begin{itemize}
            \item Replace it with a simple edge that connects the two vertices. 
            
            \item If there is a good vertex $g^{alg}_{i}$ that is adjacent to $b_e(q_i^{opt})$, remove the edge \\  $(b_e(q_i^{opt}), g^{alg}_{i})$, and add the edge  $(g^{alg}_{i}, b_e(q_{i+1}^{opt}))$.

            \item If there is not such vertex, then 
        find a good vertex $g^{alg}_{i+1}$ that is adjacent to \\ $b_e(q_{i+1}^{opt})$, remove the edge $(b_e(q_{i+1}^{opt}), g^{alg}_{i+1})$, and add the edge $(g^{alg}_{i+1}, b_e(q_{i}^{opt}))$.
        \end{itemize}

        \item Form an Eulerian tour $T^H$ in $H$.

        \item Consider sequence $s$ of the visited vertices in $T^H$. Skip every repeated bad vertex\\ $b \in V^b$ that has at least one adjacent good vertex (the first vertex is adjacent to\\ the last one in $s$) until all bad vertices are unique.

        \item Skip all repeated good vertices and output a Hamiltonian cycle.

    \end{enumerate}
    \caption{\textsc{ShortCut}($H$)}
    \label{shortcut}
\end{algorithm}

Let us first state a lemma that will help us compute the time complexity of Algorithm~\ref{algo}.

\begin{lemma}[Lemma~1 in~\cite{zhou}]
    Given a graph $G=(V,E)$ and a subset $V'$ of $V$, a minimum spanning forest of $G$ rooted at $V'$ can be computed in $O(|V|^2)$ time.
\end{lemma}

\begin{theorem}
        Algorithm~\ref{algo} runs in $O\big((3k)! 8^k \cdot n^4\big)$ time. 
\end{theorem}
\begin{proof}
    Step~1 computes $k$ violating triangles in $O(n^3)$ time. The number of $t$-partitions of permutations enumerated in Step~2 is $|V^b|! 2^{|V^b|} = O\big((3k)! 8^k\big)$. For each possible $t$-partition, Step~2(a) to Step~2(e) take in total $O(n^4)$. It is easy to see that the computation of a minimum-cost perfect matching is the most time-consuming step of Algorithm~\ref{algo}. We can compute a minimum-cost perfect matching in time $O(n^4)$ using Edmonds' algorithm~\cite{Edmonds1, Edmonds2}. Note that Algorithm~\ref{shortcut} runs in linear time. Consequently, Algorithm~\ref{algo} runs in $O\big((3k)! 8^k \cdot n^4 \big)$ time. 
\end{proof}

\subsection{Analysis of Approximation Ratio}
In this section, we prove the correctness of Algorithm~\ref{algo} and show that it has an approximation ratio of $2.5$. First, we have that for an optimal TSP solution $c(\opt) = \sum_{i=1}^{t^{opt}}c(q_i^{opt}) + \sum_{i=1}^{t^{opt}}c(p_i^{opt})$.

Since we enumerate all partitions of every permutation of bad vertices in Step~2 of Algorithm~\ref{algo}, there exists an enumeration case, where we have $t^{opt}$ bad chains which contain exactly the same bad vertices in the same order as in $\opt$. In the following, we focus on the case when the permutation of bad vertices (with the right “gaps” of good vertices) in our algorithm and $\opt$ match. For a graphic illustration of Steps~2(a) and~2(b) of Algorithm~\ref{algo} see Figure~\ref{fig:algorithm-1}.

First, we compute the cost of the cycle $C$ of bad vertices formed by Algorithm~\ref{algo} after Step~2(a). The cost of $C$ is 
\begin{equation} \label{eq_a1}
    c(C) = \sum_{i=1}^{t^{opt}}c(q_i^{opt}) + \sum_{i=1}^{t^{opt}}c(b_e(q_i^{opt}), b_s(q_{i+1}^{opt})).
\end{equation}
To prove that the cost of the cycle of bad vertices $C$ is upper bounded by $c(\opt)$ we use a lemma stated in~\cite{zhou}. In the following, we give its full proof for completeness.

\begin{lemma}[Lemma~3 in~\cite{zhou}] \label{lemma_a}
    For each $i \in [t^{opt}]$, $c(b_e(q_i^{opt}), b_s(q_{i+1}^{opt})) \le c(p_i^{opt})$, where $q_{t^{opt}+1}^{opt} = q_1^{opt}$.
\end{lemma}

\begin{proof}
    Every triangle that contains at least one good vertex satisfies the triangle inequality (all three of them between the three vertices of the triangle). Therefore, using repeatedly the triangle inequality we have that for each $i \in [t^{opt}]:$
    \begin{align*}
        c(b_e(q_i^{opt}), b_s(q_{i+1}^{opt})) \le c(b_e(q_i^{opt}), g_{i_1}^{opt}) + c(g_{i_1}^{opt} , b_s(q_{i+1}^{opt})) \\
        \le c(b_e(q_i^{opt}), g_{i_1}^{opt}) + c(g_{i_1}^{opt}, g_{i_{\ell_i^{opt}}}^{opt})+ c(g_{i_{\ell_i^{opt}}}^{opt} , b_s(q_{i+1}^{opt})) \\
        \le c(b_e(q_i^{opt}), g_{i_1}^{opt}) + c(g_{i_1}^{opt},g_{i_2}^{opt}, \dots, g_{i_{\ell_i^{opt}}}^{opt})+ c(g_{i_{\ell_i^{opt}}}^{opt} , b_s(q_{i+1}^{opt})) \\
        \le c(p_i^{opt}).
    \end{align*}
\end{proof}
From~(\ref{eq_a1}) and Lemma~\ref{lemma_a} we get that after Step~2(a): $c(C) \le c(\opt)$.

After Step~2(b), we obtain a $t^{opt}$-minimum spanning forest ($t^{opt}$-MSF) $F=(T_1,\dots,T_{t^{opt}})$ of  $G[V^g \cup \{b_e(q_i)|i \in [t^{opt}]\}]$ rooted at $\{b_e(q_i)|i \in [t^{opt}]\}$. As shown in Lemma~5 of~\cite{zhou}, it holds that $c(F) \le c(\opt)$. Here again, we give the full proof for completeness.

\begin{lemma}[Lemma~5 in~\cite{zhou}]
    $c(F) \le c(\opt)$.
\end{lemma}
\begin{proof}
    It is easy to see that a $t^{opt}$-spanning forest of  $G[V^g \cup \{b_e(q_i)|i \in [t^{opt}]\}]$ rooted at $\{b_e(q_i)|i \in [t^{opt}]\}$ can be formed by removing some edges and/or vertices from an optimal solution $\opt$. To show that, we first remove all bad vertices not in $\{b_e(q_i)|i \in [t^{opt}]\}$ and all edges incident to them from $\opt$. Then, for each $i \in [t^{opt}]$, if $q_i^{opt}$ consists of only one bad vertex, then we remove the edge connecting it with its preceding good vertex. Thus, we get the desired $t^{opt}$-spanning forest which has cost greater than or equal to the cost of the $t^{opt}$-minimum spanning forest $F$.
\end{proof}

\begin{figure}[ht]
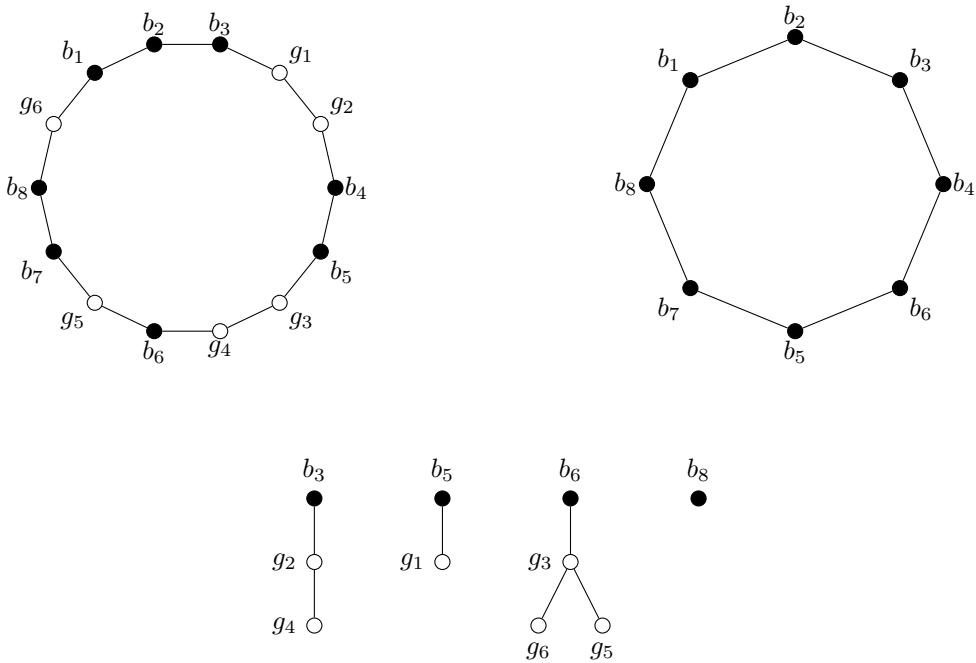

    \centering
  \begin{subfigure}{6cm}
    \centering\includegraphics[width=5cm]{figure_1}
  \end{subfigure} \hspace*{\fill}
  \begin{subfigure}{6cm}
    \centering\includegraphics[width=5cm]{figure_2}
  \end{subfigure}
  \begin{subfigure}{6cm}
    \vspace{1cm}
    \centering\includegraphics[width=6cm]{figure_3}
  \end{subfigure}
  \caption{An illustration of Steps~2(a) and~2(b) of Algorithm~\ref{algo}. $b_i$'s are bad vertices and $g_i$'s are good vertices. Top left: the cycle represents an optimal TSP solution. The algorithm “guesses” the occurrence order of bad vertices and the “gaps” between bad vertices in the optimal TSP solution, but the position of good vertices in the solution remains unknown. Top right: In Step~2(a), Algorithm~\ref{algo} creates a simple cycle $C$ of only bad vertices following their order in the optimal solution. Bottom: The algorithm computes a $4$-minimum spanning forest ($t=4$, as the number of bad chains in the optimal TSP solution) of good vertices rooted at the bad end-vertices of each chain ($b_3, b_5, b_6, b_8$). Note that in this example, the $4$-minimum spanning forest does not connect any good vertex to $b_8$.  
  }
  \label{fig:algorithm-1}
\end{figure}

Next, in Step~2(c), a connected graph $G'$ that contains all vertices $V$ of $G$ and the edges of $C$ and $F$ is created. Note that the number of odd degree vertices $O$ in $G'$ is even. Thus, there exists a perfect matching in the subgraph induced in $G$ by $O$. Note that the vertices of $O$ are either good vertices or bad vertices $\{b_e(q_i)|i \in [t^{opt}]\}$ that are roots for the spanning forest $F$, as each (bad) vertex that is not contained in $F$ has degree exactly $2$, by construction.

In Step~2(d), the algorithm computes a minimum-cost perfect matching $M$ in the subgraph induced in $G'$ by $O$. We will show that $c(M) \le c(\opt)/2$.

\begin{lemma}
     $c(M) \le c(\opt)/2$.
\end{lemma}
\begin{proof}
    Consider the optimal TSP tour (cycle) $\opt$ and the subset $O$ of vertices that have odd degree in $G'$. Note that all odd degree vertices are either good vertices or bad end-vertices $b_e(q_i^{opt})$. Every consecutive pair of bad end-vertices in $\opt$ has at least one good vertex between them and adjacent to one of them in $\opt$. Let us denote by $(o_1, o_2, \dots, o_{|O|})$ the ordered-set of odd degree vertices in $G'$ in the order of $\opt$. We will now shortcut $\opt$ to get a cycle consisting of only vertices of odd degree.
    We can write $\opt$ as follows:
    \begin{equation*}
        \opt = (o_1, u_{1_1},\dots, u_{1_{\ell_1}}, o_2, \dots, o_{|O|}, u_{|O|_{1}}, \dots,  u_{|O|_{\ell_{|O|}}}).
    \end{equation*}
    We have two cases for a pair of consecutive odd degree vertices $(o_{i}, o_{i+1})$:
    \begin{enumerate}
        \item If $o_i$ or $o_{i+1}$ is a good vertex, then we can apply the triangle inequality repeatedly using that good vertex and skip all vertices between $o_i$ and $o_{i+1}$ in $\opt$ without increasing the cost.
        \item If both $o_i$ and $o_{i+1}$ are bad vertices, then as we explained there is a good vertex $g$ between them and adjacent to $o_i$ or $o_{i+1}$. Using that good vertex we can again repeatedly apply triangle inequality and skip all vertices between $o_i$ and $o_{i+1}$ in $\opt$ without increasing the cost. 
    \end{enumerate}

    Using triangle inequality in the way we described, we can get a simple cycle of all odd degree vertices with cost upper bounded by $c(\opt)$. Now it is easy to see that we have two disjoint perfect matchings and at least one them has cost at most $c(\opt)/2$. Therefore, it also holds that $c(M) \le c(\opt)/2$.
\end{proof}

After Step~2(e) in Algorithm~\ref{algo}, a connected multigraph $H$ of all $n$ vertices is formed, where each vertex has even degree, and the total cost of all edges in $H$ is $$c(H) \le c(C)+c(F)+c(M) \le 2.5 \cdot  c(\opt).$$

\begin{figure}[ht]
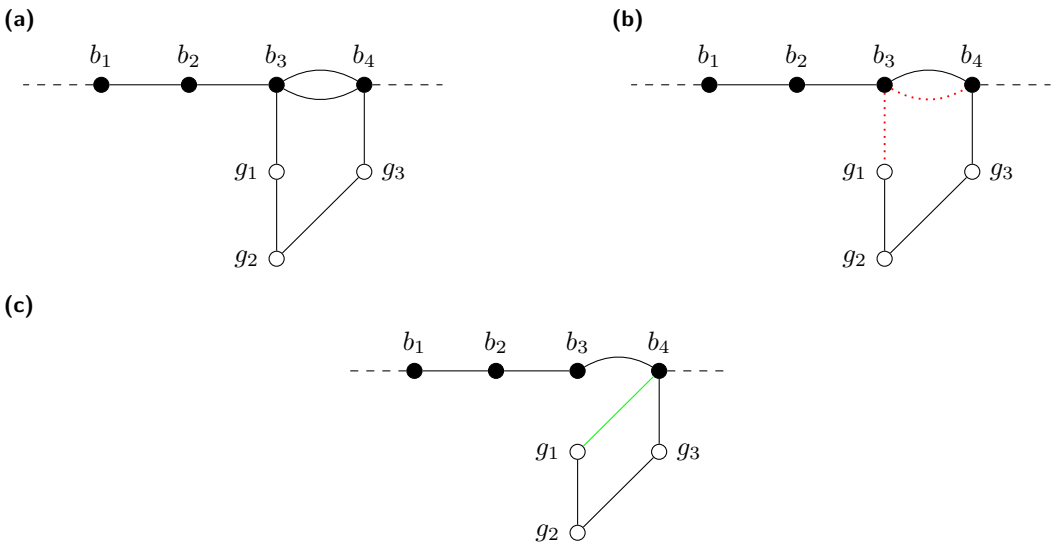

    \centering
  \begin{subfigure}[b]{6cm}
    \caption{}
    \centering\includegraphics[width=\textwidth]{figure_4} 
  \end{subfigure} \hspace*{\fill}
  \begin{subfigure}[b]{6cm}
      \caption{}
    \centering\includegraphics[width=\textwidth]{figure_5} 
  \end{subfigure}
  \begin{subfigure}[b]{\textwidth}
      \caption{}
    \centering\includegraphics[width=0.4*\textwidth]{figure_6} 
  \end{subfigure}
  \caption{An illustration of Step~1 of \textsc{ShortCut} (Algorithm~\ref{shortcut}). $b_i$'s are bad vertices and $g_i$'s are good vertices. Note that before (and after) this step, all vertices have even degree. The purpose of this step is to remove all double edges between bad vertices from the created graph. Here, we would like to remove the double edge $(b_3,b_4)$ ($b_3, b_4$ are both end vertices, see Figure~(a)). As you can see in Figure~(b), we remove one edge $(b_3, b_4)$ and the edge $(b_3, g_1)$ (red dotted edges). Finally, we connect $g_1$ with $b_4$ (green edge) ensuring that the graph is connected and all vertices have even degree. It is trivial to see that this step does not increase the total cost of the edges of the graph due to the triangle inequality in $\Delta(b_3,g_1, b_4)$.}
  \label{fig:shortcut}
\end{figure}

In the last step of the algorithm (Step~2(f)), we run subroutine $\textsc{ShortCut}(H)$ (see Figure~\ref{fig:shortcut} for an illustration). As we explained, all bad vertices have even degree in $G'$ (before matching $M$) except maybe for the bad end-vertices $b_e(q_i^{opt})$. Specifically, all bad vertices that are not end-vertices have degree exactly $2$ before Step~2(d). Thus, the only case in which a double-edge between two bad vertices can exist in $H$ is when connecting two bad end-vertices (at least one of them has to form a bad chain of a single vertex).

In Step~1 of $\textsc{ShortCut}$, for each double edge between bad vertices $(b_e(q_i^{opt}), b_e(q_{i+1}^{opt}))$ we first find a good vertex $g_i^{alg}$ that is adjacent to $b_e(q_i^{opt})$ or a good vertex $g_i^{alg}$ that is adjacent to $b_e(q_i^{opt})$. To show the existence of such a good vertex, assume that there are no good vertices adjacent to either  $b_e(q_i^{opt})$ or  $b_e(q_{i+1}^{opt})$. Then, by construction both bad vertices would have degree $2$ before Step~2(d) of Algorithm~\ref{algo} (before the matching) and a double edge $(b_e(q_i^{opt}), b_e(q_{i+1}^{opt}))$ could not occur in $H$: both vertices would be connected with an edge and each one of them would be connected with another bad vertex in the cycle (the only way that a double edge would occur is if $b_e(q_i^{opt}), b_e(q_{i+1}^{opt})$ were the only vertices in $G$, a trivial case). Therefore, there exists such a good vertex for every double edge between bad vertices.

Assume wlog that we have a good vertex  $g_i^{alg}$. In Step~1, Algorithm~\ref{shortcut} removes an edge $(b_e(q_i^{opt}), b_e(q_{i+1}^{opt}))$ and edge $(b_e(q_i^{opt}), g_i^{alg})$ and adds edge $(g_i^{alg}, b_e(q_i^{opt}))$. Using triangle inequality 
\begin{equation*}
    c(g_i^{alg}, b_e(q_i^{opt}) \le (b_e(q_i^{opt}), b_e(q_{i+1}^{opt})) + c(b_e(q_i^{opt}), g_i^{alg}),
\end{equation*}
we have that Step~1 does not increase the cost of the edges of the multigraph $H$.
The triangle inequality holds in any triangle that contains a good vertex. Note also that this step reduces the degree of one bad vertex by 2, and does not change the degree of other vertices. Then, after this step, every vertex still has even degree. Thus, in Step~2 an Eulerian tour $T^H$ in $H$ can be found. The length of $T^H$ is at most $2.5 \cdot c(\opt)$. 

Consider tour $T^H$ as a sequence $s$ of the visited vertices. Let us state a useful lemma about two occurrences of a bad vertex in $s$.

\begin{lemma}\label{lemma_bad}
    Let $b^{(1)}$ and $b^{(2)}$ be two (consecutive) occurrences of a bad vertex $b \in V^b$ in $s$. Then, there exists an adjacent good vertex $g \in V^g$ in $s$ to one of the two occurrences of $b$.
\end{lemma}
\begin{proof} 
    By construction we have that each bad vertex is adjacent to at most $3$ bad vertices in $H$ (two in the cycle $C$ of bad vertices, and at most one in the matching $M$) and that there are no double edges between two bad vertices. 
    We distinguish the following cases:
    \begin{itemize}
        \item If there is only one vertex between $b^{(1)}$ and $b^{(2)}$, it cannot be bad as there are no double edges connecting bad vertices, so it is a good one and the statement occurs.

        \item If there are at least two  vertices between $b^{(1)}$ and $b^{(2)}$ then, since every bad vertex is adjacent to at most $3$ bad vertices in $H$, we have that there exists at least one good vertex adjacent  to either $b^{(1)}$ or $b^{(2)}$.
    \end{itemize}
    Therefore, in all cases the statement of the lemma is true.
\end{proof}

Due to Lemma~\ref{lemma_bad} we can skip every bad vertex that appears in $s$ more than once without increasing the cost of the tour using triangle inequality (Step~3). In Step~4, we remove the good vertices that appear more than once in $s$ using again the triangle inequality. Consequently, we get a simple cycle of all $n$ vertices (TSP tour) with cost $c(\alg) \le 2.5 \cdot c(\opt)$ and state the following theorem.

\begin{theorem}
    Algorithm~\ref{algo} has an approximation ratio of $2.5$.
\end{theorem}

\section{Conclusion}
We studied the TSP problem parameterized by the number of triangles in which the edge costs violate the triangle inequality and gave a $2.5$-approximation algorithm that runs in FPT time, improving the approximation ratio of $3$ in~\cite{zhou}. We believe that this line of work for TSP is particularly interesting since it bridges the gap between the non-metric and metric case. A natural question is whether one can further improve the approximation factor, getting closer to the approximation factor of metric TSP. Another interesting research direction is to consider the second parameter $k'$, introduced in~\cite{zhou}, which measures the minimum number of vertices whose removal turns the input instance into a metric one. It is an open question if there exists a constant approximation FPT algorithm parameterized by $k'$.

\bibliography{bibliography}

\end{document}